\documentclass{article}

\usepackage{xspace}
\usepackage{verbatim}
\usepackage{color}
\usepackage{graphicx}
\usepackage{tabularx}

\usepackage{amsmath}
\usepackage[showonlyrefs]{mathtools}
\usepackage{mathstyle}
\usepackage{empheq}
\usepackage{amstext,amssymb,amsfonts}
\usepackage{fullpage}
\usepackage{nicefrac}
\usepackage{bm}

\usepackage{todonotes}

\newcounter{mynotes}
\usepackage{textcomp,setspace}

\usepackage{nameref}
\usepackage[pagebackref,colorlinks,linkcolor=blue,filecolor = blue,citecolor = blue, urlcolor = blue, hyperfootnotes=false]{hyperref}
\usepackage{color}
\usepackage{cleveref}

\usepackage{amsthm}
\usepackage{thmtools}

\declaretheorem[within=section]{theorem}
\declaretheorem[sibling=theorem]{corollary}

\declaretheorem[sibling=theorem]{lemma}

\declaretheorem[sibling=theorem]{definition}

\declaretheorem[sibling=theorem]{remark}

\newcounter{termcounter}
\renewcommand{\thetermcounter}{\Alph{termcounter}}
\crefname{term}{term}{terms}
\creflabelformat{term}{\textup{(#2#1#3)}}

\makeatletter
\def\term{\@ifnextchar[\term@optarg\term@noarg}
\def\term@optarg[#1]#2{%
  \textup{(#1)}%
  \def\@currentlabel{#1}%
  \def\cref@currentlabel{[][2147483647][]#1}%
  \cref@label[term]{#2}}
\def\term@noarg#1{%
  \refstepcounter{termcounter}%
  \textup{(\thetermcounter)}%
  \cref@label[term]{#1}}
\makeatother

\newcommand{\ignore}[1]{}

\newcommand{\poly}{\mathrm{poly}}

\definecolor{DSred}{rgb}{1,0,0}

\renewcommand{\leq}{\leqslant}
\renewcommand{\geq}{\geqslant}
\renewcommand{\ge}{\geqslant}

\renewcommand{\epsilon}{\varepsilon}
\newcommand{\eps}{\epsilon}

\newcommand{\R}{\mathbb{R}}

\newcommand{\cE}{\mathcal E}

\newcommand{\cN}{\mathcal N}

\usepackage{caption}
\usepackage{subcaption}

\renewcommand{\cref}{\Cref}

\allowdisplaybreaks[1]

\begin{document}

\title{How friends and non-determinism affect opinion dynamics}

\author{Arnab Bhattacharyya \qquad\qquad\qquad Kirankumar Shiragur\\[1em]
Department of Computer Science \& Automation\\
Indian Institute of Science\\
\href{}{\texttt{\{arnabb, kirankumar.shiragur\}@csa.iisc.ernet.in}}
}


\date{\today}
\maketitle

\begin{abstract}
The {\em Hegselmann-Krause system} ({\em HK system} for short) is one of the most popular models for the dynamics of opinion formation in multiagent systems. Agents are modeled as points in opinion space, and at every time step, each agent moves to the mass center of all the agents within unit distance. The rate of convergence of HK systems has been the subject of several recent works. In this work, we investigate two natural variations of the HK system and their effect on the dynamics. In the first variation, we only allow pairs of agents who are friends in an underlying social network to communicate with each other. In the second variation, agents may not move exactly to the mass center but somewhere close to it. The dynamics of both variants are qualitatively very different from that of the classical HK system. Nevertheless, we prove that both these systems converge in polynomial number of non-trivial steps, regardless of the social network in the first variant and noise patterns in the second variant.
\end{abstract}

\section{Introduction}

The dynamics of opinions in society is a very intricate and intriguing process. Especially in today's world, with the pervasive infiltration of social networks like Facebook and Twitter that allow quick broadcasts of opinions, phenomena which were once wild speculations by philosophers, such as the viral spread of memes \cite{Dawkins76}, are now easily observed and quantified. This acceleration of social processes is turning sociology into a quantitative science, where concrete models for social phenomena can be proposed and rigorously tested.

Sociologists have long identified several different processes that determine opinion dynamics, most notably, {\em normative} and {\em informational} \cite{DG55}. Normative influence refers to the influence that causes people to conform to a group's social norms. On the other hand, informational influence refers to the way people acquire the opinions of others, driven by the assumption that neighbors possess information about a situation. Informational influence is especially relevant in the context of understanding how opinions {\em change}, and it arguably is the dominant process for determining trends in, say, fashion, mobile phones and music.

In this work, we focus on quantitative models for informational influence. Such a model should specify how an individual agent updates its opinion using information learned from its ``neighbors". By now, this area has been heavily studied; see \cite{Jackson08} for a survey. The most basic such model is the classic DeGroot model \cite{deGroot74, French56, Harary59} where each agent's opinion is a real number between 0 and 1, and at every time step, each agent moves to some weighted average of its neighbors' positions, where the neighbors are determined according to an unchanging undirected graph. Such a system always reaches consensus, contrary to the existence of polarized states in society. Another classic model is the {\em voter model} \cite{CS73, HL75}. Here, there is an unchanging directed graph among the agents, and at every time step, a random agent selects a random neighbor and adopts the neighbor's opinion as its own. Again, such a system always reaches a consensus, and coalescing random walks \cite{DKS91} can be used to bound the convergence time.

In order to explain why consensus doesn't always arise in the real world, one can posit the presence of {\em stubborn agents}, agents which never change their own opinions (though they may certainly influence others). More generally, multiple studies \cite{Asch55, DG55, LRL79} have confirmed that even when agents are not stubborn, they usually have a {\em conformity bias}, i.e., they assign more weight to opinions that are already close to their own. This notion gives rise to the definition of {\em influence systems} or {\em flocking models}. 

\begin{figure*}[t]
\centering
\begin{subfigure}[b]{0.3\textwidth}
\includegraphics[width=\textwidth, height=2cm]{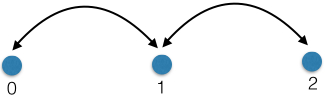}
\caption{No freezing time}
\label{fig:nofrz}a
\end{subfigure}
\begin{subfigure}[b]{0.3\textwidth}
\includegraphics[width=\textwidth, height=2.5cm]{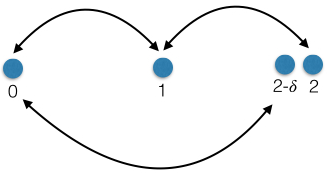}
\caption{Time to converge can be arbitrarily long}
\label{fig:initdep}
\end{subfigure}
\begin{subfigure}[b]{0.3\textwidth}
\includegraphics[width=\textwidth, height=4cm]{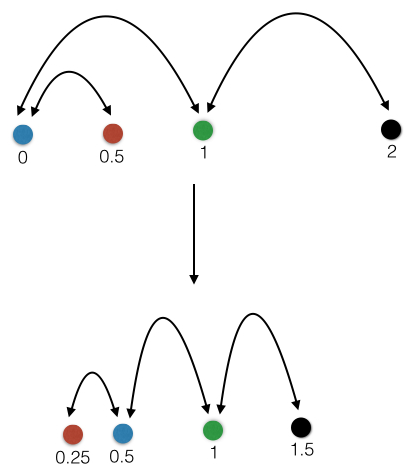}
\caption{Order of agents can change}
\label{fig:noorder}
\end{subfigure}
\caption{Properties of social HK systems}
\label{fig:sochk}
\end{figure*}

The most popular such flocking model is the {\em Hegselmann-Krause} system \cite{Krause97, HK02}, and this is the system that we focus on here. In its simplest incarnation, the system consists of $n$ agents placed on the real line, with the $i$'th agent at $x_t(i)$ at time $t$, and at every time step $t \geq 0$, the positions update as follows for all $i \in [n]$ synchronously:
\begin{equation}\label{eqn:hkupdate}
x_{t+1}(i) = \frac{1}{|\cN_{t}(i)|} \sum_{j \in \cN_{t}(i)} x_{t}(j)
\end{equation}
where $\cN_{t}(i) = \{j : |x_t(j) - x_t(i)|\leq 1\}$. Here, $1$ is  the {\em confidence bound}, as each agent only has confidence in those agents which are within this bound. 

The HK system  has  become quite popular as a mathematically clean and simple formalization of an endogenous dynamical system that captures interesting qualitative properties, such as polarization and conformity bias, found in opinion dynamics\footnote{{\small Moreover, the HK system can also model robotic rendezvous problems in plane and space \cite{BCM09}.}}. The expectation is that a systematic understanding of the dynamics of the HK system can lend insight into more detailed models and, hopefully, (some aspects of) reality. Indeed, convergence for the HK system is immediate, and the time complexity needed for the system to freeze is known upto a linear factor in $n$ (the best upper bound is $O(n^3)$ \cite{BBCN13,MT13} and the best lower bound $\Omega(n^2)$ \cite{WH14}). Note that it is not a priori clear that there exists a bound which depends only on $n$ and is independent of the agents' initial positions. 

But perhaps surprisingly, it has turned out that changing the model in even very simple ways leads to problems which we cannot handle mathematically. For example, one of the most common variations is to let the confidence bound depend on the agent \cite{Lorenz10, MB12}. That is, suppose $\cN_t(i) = \{j : |x_t(j) - x_t(i)| \leq r_i\}$ for some $r_i \geq 0$. This is called the {\em heterogeneous Hegselmann-Krause} system. A rigorous proof of its convergence is still missing, although convergence seems clear from simulations! This situation clarifies that we need to develop new technical tools in order to understand the dynamics of influence systems. Indeed, for general influence systems, Chazelle, in an impressive sequence of papers \cite{Chazelle11, Chazelle12}, developed a new algorithmic calculus to show that general influence systems are almost surely asymptotically periodic. However, these general results do not imply convergence for the specific case of heterogeneous HK.

\subsection{Our Contributions}
In this work, we rigorously study the convergence behavior of two variations of the HK model. Apart from their intrinsic interest from a sociology perspective, their analysis seem to pose similar mathematical difficulties as heterogeneous HK. Nevertheless,  we show that for both these systems, we can develop tools to understand their convergence behavior. Specifically, we study the following two models:

\begin{itemize}
\item
\textbf{{Social Hegselmann-Krause}}: 
One criticism of the Hegselmann-Krause model is that it only considers informational influence and ignores normative effects altogether. It ignores the fact that individuals belong to groups, and generally, information exchange only occurs between individuals in the same group. To model this fact, we assume that there exists an underlying social network such that two agents can interact with each other only when there is an edge between the two in this graph.

We formally define the {\em social Hegselmann-Krause} system as follows: given an undirected graph $G$ on $n$ nodes and a collection of $n$ agents initially at positions $x_0(1), \dots, x_0(n) \in \R$ respectively,   the social HK system updates the agents' positions synchronously for $t \geq1$ according to Equation (\ref{eqn:hkupdate}) where\footnote{{\small $E(G)$ denotes the edge set of $G$.}} $\cN_t(i) = \{j : (i, j) \in E(G) \text{ and } |x_t(j) -x_t(i)| \leq 1\}$. 

The social HK model differs in some very basic ways from the usual HK model, as shown in \cref{fig:sochk}. First of all, it is no longer true that the agents freeze after some time; agents can keep moving by some tiny amount indefinitely as the example in \cref{fig:nofrz} shows. We might hope though that, for every $\eps > 0$, after some time bound that depends on $n$ and $\eps$, the points stay within intervals of $\eps$. Even this is not true as the situation in \cref{fig:initdep} of the panel shows; there,  by making $\delta >0$ arbitrarily small, the agent initially at $0$ can take arbitrarily long to ``see" the agent initially at $2-\delta$. Finally, unlike the usual HK system, the agents do not preserve their order, as is clear from \cref{fig:noorder}.

As far as we know, we are the first to formally study the convergence behavior of the social HK model. Fortunato \cite{Fortunato05} also investigated the same system but to address a very different problem: Given a random initial configuration of agents in the interval $[0,1]$, what is the minimum confidence bound that ensures that the agents come to a consensus when the dynamics is that of social HK on a random graph of degree $d$? Fortunato's empirical result\footnote{{\small A similar empirical  result \cite{Fortunato04} has been rigorously proven \cite{ LU07} for the closely related Deffuant-Weisbuch model~\cite{WDAN02} on a social network.}} is that this threshold is $\sim 0.2$ when $d=\omega(1)$ and is $0.5$ when $d$ is constant. Perhaps in the style typical of physicists, their work focused on the equilibrium outcome, whereas we study the transient behavior. 

Given {\em $\eps$}, let us call a step of the dynamical system {\em $\eps$-non-trivial} if at least one pair of interacting agents is separatated by distance at most  {\em $\eps$}.
\begin{theorem}\label{thm:socmain}
Given an arbitrary initial configuration of $n$ agents evolving according to the social HK model defined by an arbitrary graph, the number of $\eps$-non-trivial steps in the dynamical system is $O(n^5/\eps^2)$. 
\end{theorem}
Chazelle's result \cite{Chazelle11} implies an $n^{O(n)}$ bound for this system whereas our bound is polynomial. 

We also show that the same bound holds when the social network itself changes with time, provided its evolution follows certain constraints. Informally, we require that if two agents interact at time $t$ and they are within each other's confidence bound at time $t+1$, then they should keep on interacting at time $t+1$. In particular, if edges are never deleted from the social network, then only polynomial number of non-trivial steps take place

\item
\begin{figure}
\begin{center}
\includegraphics[width=0.5\textwidth, height=4cm]{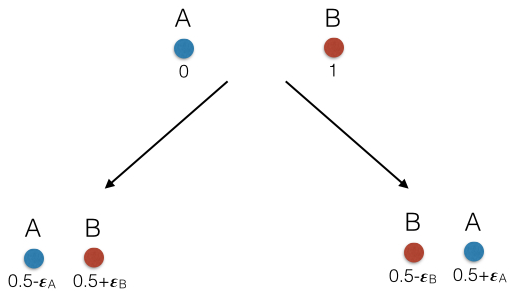}
\end{center}
\caption{Order not preserved by non-deterministic dynamics}
\label{fig:nondet}
\end{figure}

\textbf{Non-deterministic Hegselmann-Krause}: A different criticism of the HK system is that it is very rigid in that each agent must move to exactly the  mass center of the agents within unit distance. In particular, if two agents have the same set of agents within unit distance, then they move to exactly the same opinion at the next time step (and stay together thereafter). This is clearly not very realistic, as the effects of chance and variation are not taken into account.

To address these issues, for any fixed $\eps \in [0,1]$, we formally define the {\em $\eps$-non-deterministic Hegselmann-Krause} system. The system again consists of $n$ agents placed on the real line, with the $i$'th agent at $x_t(i)$ at time $t$, and at every time step $t \geq 0$, the positions update for all $i\in[n]$ synchronously as:
\begin{equation}
x_{t+1}(i) = x_t(i) + (1 + \eps_{i,t}) \sum_{j \in \cN_t(i)} \frac{x_t(j)-x_t(i)}{|\cN_t(i)|}
\end{equation}
where $\eps_{i,t} \in [-\eps, \eps]$ for every $i$ and $t$ independently, and $\cN_i(t) = \{j: |x_t(j) - x_t(i)|\leq 1\}$. Note that we term the system ``non-deterministic" instead of ``noisy", because the $\eps_{i,t}$'s are not assumed to be random. In fact, they can be entirely arbitrary values in $[-\eps,\eps]$, even chosen by an adversary depending on the current state. However, note that consensus remains a fixed point of this dynamics.

The dynamics of non-deterministic HK is quite different from that of HK. As \cref{fig:nondet} shows, the order of agents can change even when there are two agents and $\eps$ is arbitrarily small. Also, the fact that agents do not coalesce together once they see the same set of agents complicates the  system's behavior significantly. The work here is the first to handle such a general type of non-determinism in doing convergence analysis. In \cite{BBCN13}, a one-sided version of noise was discussed but there, the authors could adapt the argument to bound the convergence time of exact HK in a simple way. In contrast, we are not able do so here.
Our main result is:
\begin{theorem}\label{thm:ndmain}
Suppose $\eps < \frac{1}{4n^2}$.  Starting with an arbitrary initial arangement of $n$ agents evolving according to  the non-deterministic HK model, the number of steps before which all the agents are confined within intervals of length $\rho$ thereafter is $O(n^4 + \log(1/\rho)/\log n)$.
\end{theorem}

Again, in this case too, Chazelle's general result on bidirectional influence systems \cite{Chazelle11} implies an $n^{O(n)}$ bound whereas we show a polynomial bound.
\end{itemize}

\subsection{Our Techniques}
\paragraph{Social HK.}
In \cite{BBCN13}, the classical HK system is shown to converge in $O(n^3)$ steps by simply showing that the diameter of the system must shrink by a significant amount at each timestep, unless the leftmost and rightmost agents have already frozen. It is not clear how to extend this proof approach to the social HK system, because as \cref{fig:initdep} shows, the leftmost agent may not be frozen but may take an arbitrary long time before it moves by a significant amount. Therefore, we must search for some other energy function, which is not the diameter but which is also decreasing.

In fact, such an energy function has already been introduced \cite{RMF08}. For a given configuration $x = (x(1), \dots, x(n))$, define $\cE(x) = \sum_{(i,j)} \min(|x(i)-x(j)|^2, 1)$. The fact that this energy is decreasing is an important ingredient in several previous works \cite{BHT07, BHT09, BBCN13, EB14}. We generalize this energy function to the social HK setting as: $\cE(x) = \sum_{i\sim_t j} |x(i) - x(j)|^2 + \sum_{i \not\sim_t j} 1$, where $i \sim_t j$ means that agents $i$ and $j$ interact at time $t$. 

 To give a lower bound on the decrease of this energy function, we use an elegant approach introduced recently by Martinsson \cite{Martinsson15}. Martinsson relates the decrease to the eigenvalue gap of the communication graph, which is well-known to be $\geq 1/\poly(n)$. We show that the same approach applies to the social HK system also. We get that the energy decrement in every $\eps$-non-trivial step is $\Omega(\frac{\eps^2}{n^3})$, and as the energy of any initial configuration lies between 0 and $O(n^2)$, this gives an upper bound of $O(\frac{n^5}{\eps^2})$ on the number of $\eps$-non-trivial steps. This approach continues to apply even when the social network graph is changing under some natural restrictions.

\ignore{
For social HK model, most of the techniques used to prove an upper bound for convergence of the HK model don?t carry forward directly. In this paper, we show that for social HK model one cannot hope to bound the convergence time since it can be arbitrarily 
large. To deal with this problem we ignore the microscopic movements and just upper bound the number of non-trivial steps, which will be made formal later.\\
The convergence for this model is assured from \cite{Cha11}, where in Chazelle proves convergence and upper bounds the number of non-trivial steps for any bi-directional influence systems.  The upper bound we get from \cite{Cha11} on the number of non-trivial steps for this model is $n^{O(n)}$. In this paper, we bring down this upper bound to \emph{poly(n)}, which is a drastic improvement. To do so, we define an energy function $\mathcal{E}(x_t)$, which we show to be non-increasing in $t$. The technique we use to estimate the energy decrement in a step in the social HK model is similar to the one presented in \cite{Martinsson15}. We show that the energy decrement in every non-trivial step is $\Omega(\frac{\epsilon}{n^3})$, where $\epsilon \in (0,1]$ is any arbitrary fixed number and the energy $\mathcal{E}(x_0)$ of any initial configuration lies between 0 and $O(n^2)$ giving an upper bound of $O(\frac{n^5}{\epsilon})$ on the number of non-trivial steps.}

\paragraph{Non-deterministic HK.}
To analyze the non-deterministic HK system, we use a mixture of geometric and algebraic techniques similar to the ones presented in \cite{BBCN13}. However, the technical details are significantly more involved. The basic approach in \cite{BBCN13} is to look at the neighbor immediately to the right of the leftmost agent and analyze its influence. However, in our case, we need to partition the neighbourhood of the leftmost agent into two subsets and treat the subsets collectively. One might also notice that proving a lower bound on the movement of the leftmost agent as in \cite{BBCN13} is not enough, as the order of positions is not preserved and the leftmost agent may change over time. To overcome this difficulty, in our work we lower bound the amount by which the diameter (difference between the rightmost and the leftmost agents, whose identities change with time) shrinks within $n$ time steps, or else, we show that some agents separate out leaving behind sub-systems with smaller number of agents, where they both converge independently thereafter.

\section{Social HK model}\label{sec:soc}
We reformulate the social HK model in the multidimensional setting. This is very natural when an opinion consists of positions along multiple axes instead of just one. 

Let $x_t(i) \in \R^d$ (for $d \geq 1$) be the position of the $i$th agent at time $t$, and let $G$ be a fixed undirected graph on $n$ nodes. The dynamics is given by the following equation:
\begin{equation} \label{e1}
x_{t+1}(i) = \frac{1}{|\cN_t(i)|} \sum_{j \in N_t(i)} {x_t(j)}
\end{equation}
where $\cN_t(i) = \{j : (i,j) \in E(G) \text{ and }\|x_t(j) - x_t(i)\|_2 \leq 1\}$. We denote the new state at time $t+1$ by:
\begin{equation} \label{e2}
x_{t+1} = P_t x_t
\end{equation}
where $P_t$ is a row-stochastic matrix. As mentioned in the Introduction, our proof follows the same line as the recent analysis by Martinsson \cite{Martinsson15} but with some twists due to the presence of the social network graph.

For any configuration $x = (x(1), x(2), \dots, x(n))$ of $n$ agents, define the {\em communication graph} $C_x$ so that two nodes $i$ and $j$ are adjacent exactly when $(i,j) \in E(G)$ and $\|x(i) - x(j)\|_2 \leq 1$. That is, the communication graph is now the conjunction of the social network $G$ and the standard HK communication graph.
Also, for a configuration $x$ of $n$ agents, we define its {\em energy} as\\
\begin{equation} \label{e3}
\cE(x) = \sum_{(i,j) \in E(C_x)}{\|x(i)-x(j)\|_2^2} + \sum_{(i,j) \notin E(C_x)}{1}
\end{equation}
Note that the energy of any configuration lies between 0 and $n^2$. For the standard HK system, a very useful fact is that the energy $\cE(x_t)$ is non-increasing in $t$; see Theorem 2 of \cite{RMF08} for a proof. In fact, this fact is the driver for the bound on the freezing time of multidimensional HK found in \cite{BBCN13}. Our proof shows that the same energy decreases over time for the social HK system also. 

For a given state $x$ and for any ordered pair $(i,j) \in [n]^2$, we say that $(i, j)$ is {\em active} if $(i,j) \in E(C_x)$. We consequently define the {\em active part of the energy of $x$} as\\
\begin{equation} \label{e4}
\cE_{act}(x) = \sum_{(i,j) \text{ active}}{\|x(i)-x(j)\|_2^2} =\sum_{(i,j) \in E(C_x)}{\|x(i)-x(j)\|_2^2}
\end{equation}

Now, let $\{x_t\}$ be a sequence of configurations evolving according to (\ref{e2}). For simplicity of notation, let $E_t$ denote the edge set of the communication graph $C_{x_t}$. 
\begin{lemma}\label{lem:init}
\begin{align*}
\sum_{(i,j) \in E_{t+1}}&{\|x_{t+1}(i)-x_{t+1}(j)\|_2^2} + \sum_{(i,j) \notin E_{t+1}}{1} \\
&\leq \sum_{(i,j) \in E_{t}}{\|x_{t+1}(i)-x_{t+1}(j)\|_2^2} + \sum_{(i,j) \notin E_{t}}{1}
\end{align*}
\end{lemma}
\begin{proof}
There are four cases to look at:
\begin{enumerate}
\item
\textbf{$\bm{(i,j) \in E_t}$ and $\bm{(i,j) \in E_{t+1}}$}\\
In this case. we are adding the same term ($\|x_{t+1}(i)-x_{t+1}(j)\|_2^2$) to both LHS and RHS. 
\item
\textbf{$\bm{(i,j) \notin E_t}$ and $\bm{(i,j) \notin E_{t+1}}$}\\
In this case too, we are adding the same term (1) to both LHS and RHS. 
\item 
\textbf{$\bm{(i,j) \in E_t}$ and $\bm{(i,j) \notin E_{t+1}}$}\\
In this case, note that $\|x_{t+1}(i)-x_{t+1}(j)\|_2^2>1$, because otherwise $(i,j) \notin E(G)$ which contradicts the fact that $(i,j) \in E_t$.  Hence in this case, we are adding a greater term to the RHS ($\|x_{t+1}(i)-x_{t+1}(j)\|_2^2>1$) than to the LHS (1). 
\item
\textbf{$\bm{(i,j) \notin E_t}$ and $\bm{(i,j) \in E_{t+1}}$}\\
Since $ (i,j) \in E_{t+1}$, $\|x_{t+1}(i)-x_{t+1}(j)\|_2^2 \leq 1$. Hence in this case too we are adding a term (1) to RHS which is at least the term ($\|x_{t+1}(i)-x_{t+1}(j)\|_2^2 \leq 1$) added to the LHS. 
\end{enumerate}
As the inequality is true term-wise, we have  LHS $\leq$ RHS.
\end{proof}

\newtheorem{proposition}{Proposition}
\begin{proposition}[Proposition 2.2 in \cite{Martinsson15}]
For each $t \geq 0$, let
$$\lambda_t = \max \{  |\lambda| : \lambda \neq 1 \text{ is an eigenvalue of } P_t \}.$$
Then:
\begin{equation} \label{e6}
\cE(x_t)-\cE(x_{t+1}) \geq (1-\lambda_t^2)\cE_{act}(x_t)
\end{equation}
\end{proposition}
\begin{proof}
We reproduce the proof from \cite{Martinsson15} for completeness.
Let $A_t$ denote the adjacency matrix of $C_{x_t}$, and let $D_t$ denote its degree matrix, that is, the diagonal matrix whose $(i,i)$'th element is given by the degree of $i$(Recall that every vertex in $C_{x_t}$ has an edge to itself). Observe that $P_t = D_t^{-1}A_t$. Recall $E_t= E(C_{x_t})$. We have:
\begin{align*}
\cE(x_t) &= \sum_{(i,j) \in E_t}{||x_t(i)-x_t(j)||_2^2} + \sum_{(i,j) \notin E_t}{1}\\
&= 2\text{Tr}(x_t^T(D_t-A_t)x_t) + \sum_{(i,j) \notin E_t}{1}
\end{align*}
where Tr($\cdot$) denotes trace. Here, we interpret $x_t$ as an $n \times d$ matrix.

Now consider:
\begin{align*}
\cE(x_{t+1}) &= \sum_{(i,j) \in E_{t+1}}{\|x_{t+1}(i)-x_{t+1}(j)\|_2^2} + \sum_{(i,j) \notin E_{t+1}}{1} \\
&\leq \sum_{(i,j) \in E_{t}}{\|x_{t+1}(i)-x_{t+1}(j)\|_2^2} + \sum_{(i,j) \notin E_{t}}{1} &&\text{(by \cref{lem:init})}\\
& = 2Tr(x_{t+1}^T(D_{t}-A_{t})x_{t+1}) + \sum_{(i,j) \notin E_t}{1}\\
&= 2Tr((D_t^{-1}A_tx_t)^T(D_{t}-A_{t})D_t^{-1}A_tx_t) + \sum_{(i,j) \notin E_t}{1}
\end{align*}
Hence, it suffices to show that\\
\begin{align} \label{e7}
Tr((D_t^{-1}&A_tx_t)^T(D_{t}-A_{t})D_t^{-1}A_tx_t) \nonumber\\
&= Tr(x_t^TA_tD_t^{-1}(D_{t}-A_{t})D_t^{-1}A_tx_t)\nonumber\\
&\leq \lambda_t^2Tr(x_t^T(D_t-A_t)x_t)
\end{align}
Let $y_t = D^{\frac{1}{2}}x_t$ and $B = D_t^{\frac{1}{2}}P_t D_t^{-\frac{1}{2}} = D_t^{-\frac{1}{2}}A_tD_t^{-\frac{1}{2}}$. It is straightforward to show that (\ref{e7}) simplifies to\\
\begin{equation} \label{e8}
Tr(y_t^TB_t(I-B_t)B_ty_t) \leq \lambda_t^2Tr(y_t^T(I-B_t)y_t)
\end{equation}
This inequality follows by standard linear algebra.
\end{proof}
The following is a standard result in spectral graph theory:
\begin{proposition}\label{prop22} 
For any $t \geq 0$, we have
\begin{equation} \label{e10}
| \lambda_t| \leq 1- \frac{1}{n^2\text{diam}(C_{x_t})}
\end{equation}
where $diam(C_{x_t})$ denotes the graph diameter of $C_{x_t}$. If $C_{x_t}$ is not connected, we interpret $diam(C_{x_t})$ as the largest diameter of any connected component of $C_{x_t}$.
\end{proposition}
We are now ready to prove our main theorem. Observe that $\cE_{act}(x_t) > \eps^2$ whenever the $t$'th step is $\eps$-non-trivial. \cref{thm:socmain} is therefore immediately implied by the result below.
\begin{theorem}
For any $\eps > 0$, given a social HK system with $n$ agents in $\R^d$, there are $O(n^5/\eps)$ values of $t$ for which $\cE_{act}(x_t) > \eps$. 
\end{theorem}
\begin{proof} Now given any initial configuration, we have $\text{diam}(C_{x_t}) \leq n$ (or else, we can decompose the system into independent subsystems and analyze each separately). Applying Proposition \ref{prop22}, it follows that the energy decrement in each step with $\cE_{act}(x_t) > \eps$ is $\Omega(\frac{\epsilon}{n^3})$ and such steps can hence occur at most $O(\frac{n^5}{\epsilon})$ times.\\
\end{proof}

\subsection{Changing social network}\label{sec:friendly}
One may ask what happens to the convergence rate when the social network itself evolves with time and is not fixed. Let $G_t$ denote the social network graph at time $t$. To make sure that the above proof carries through, we need to suitably restrict the evolution of $G_t$.

Given a sequence of configurations $x_t = (x_t(1), \dots, x_t(n))$, we again have the communication graph $C_{x_t}$ where two nodes $i$ and $j$ are adjacent if $(i,j) \in E(G_t)$ and $\|x_t(i) - x_t(j)\|_2 \leq 1$. As before, let $E_t = E(C_{x_t})$.\\[-1em]  
\begin{definition}
Call a social HK system defined by a sequence of time-varying social networks $G_t$ as {\em friendly} if it is the case that whenever $(i,j) \in E_t$ and $\|x_{t+1}(i) - x_{t+1}(j)\|_2 \leq 1$, $(i,j) \in E(G_{t+1})$ (and hence, $(i,j) \in E_{t+1})$).
\end{definition}
In other words, in a friendly HK system, if two agents interact at time $t$, and they stay within distance $1$ in the next time step, then they keep interacting with each other at time $t+1$. Note that the evolution of $G_t$ may be endogenous (i.e., depend on the states $x_t$). We observe that under this natural condition of friendliness, the above proof goes through without any changes.
\begin{theorem}
For any $\eps > 0$, given a friendly social HK system with $n$ agents in $\R^d$, there are $O(n^5/\eps)$ values of $t$ for which $\cE_{act}(x_t) > \eps$.
\end{theorem}
\begin{proof}
The only part of the proof which needs a second look is case 3 in the proof of \cref{lem:init}. Now, the definition of friendliness ensures that if $(i,j) \notin E_{t+1}$ and $(i,j) \in E_t$, then $\|x_{t+1}(i) - x_{t+1}(j)\|_2 > 1$. 
\end{proof}
Note that without the friendliness assumption, Chazelle \cite{Chazelle11} shows a bound of $n^{O(n)}$ for the number of non-trivial steps. We conjecture that the friendliness assumption is necessary for a polynomial bound.

\subsection{Experimental Results}\label{sec:exp}
\begin{figure}
\begin{center}
\includegraphics[width=0.7\textwidth]{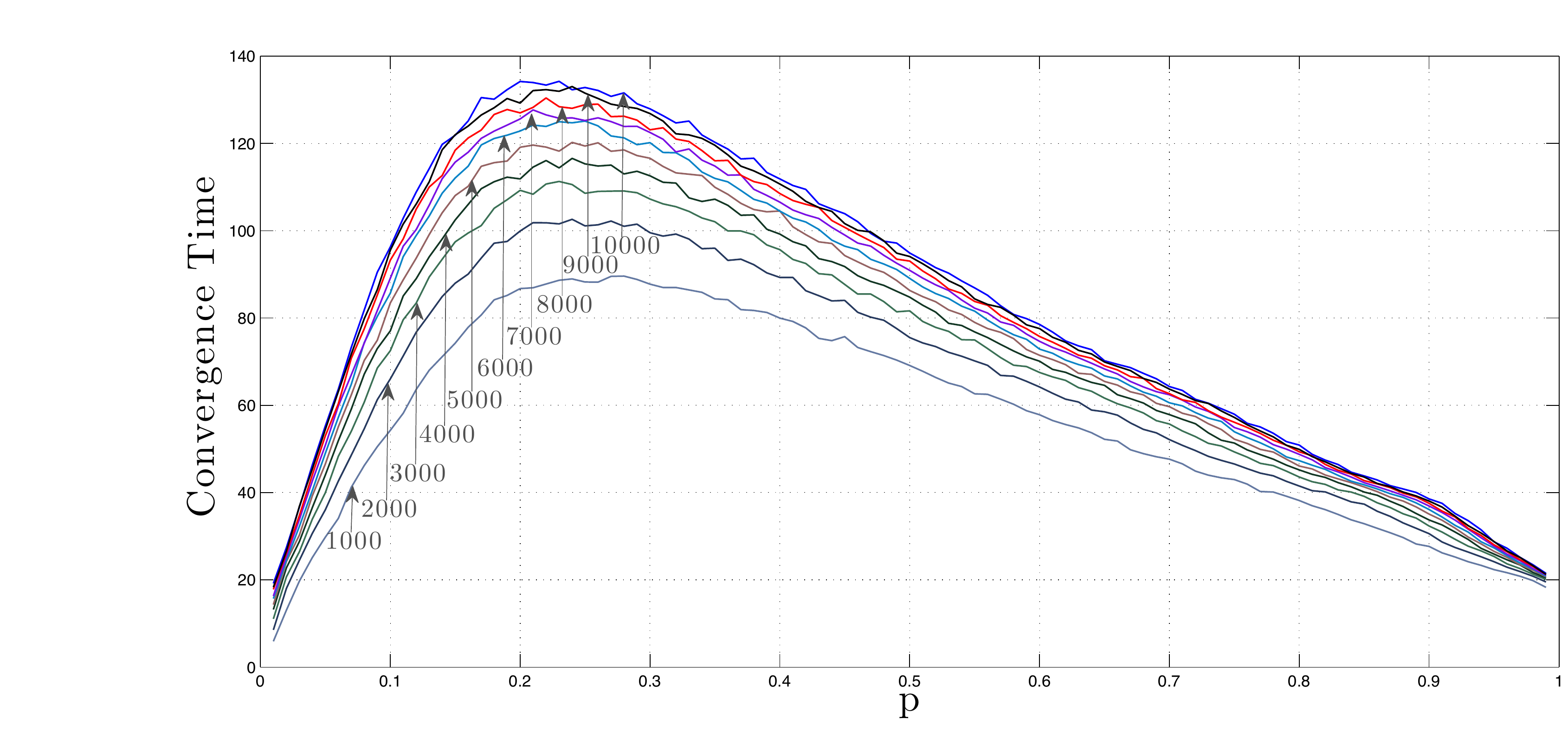}
\end{center}
\caption{{\small The plot shows the convergence time for social HK systems when the initial states of the $n$ agents are uniform in $[1,n]$ and the social network is the random graph $G(n,p)$ of edge density $p$. $n$ was varied from $1000$ to $10000$ as shown, and $p$ was varied from $0.01$ to $1$ in steps of size $0.01$. See text for discussion of ``convergence time''.}}
\label{fig:exp}
\end{figure}
Our analysis above is very general in the sense that our bound for the number of non-trivial steps does not depend on the structure of the social network. Is it true that some social network structures allow faster convergence than others?  As a first cut at this question, we explore how the edge density of the social network affects dynamics. Let $G(n,p)$ be the Erd\H{o}s-Renyi random graph on $n$ nodes, where each pair of nodes is an edge with probability $p$ independently. How does $p$ change the time needed to converge? We study this question when the initial positions of the agents are uniform in the interval $[1,n]$. \cref{fig:exp} summarizes the results of our computer experiments.

We need to clarify what we mean by ``convergence time'' in the figure. As \cref{fig:initdep} shows, the time needed to converge can, in general, be arbitrarily long. However, it seems that if the initial positions of the agents are chosen randomly, such pathological cases occur with probability $0$, so that after a finite time, all agents are confined to an interval of length $10^{-6}$ ever after. We do not have a rigorous proof of this claim, but all our simulations support it. In view of this, the notion of convergence time we used to arrive at \cref{fig:exp} was, for any given $n$ and $p$, the least time $t$ at which the sum of the movement of all the $n$ agents is less than $10^{-6}$, averaged over $1000$ random initializations of the agents and random graphs from $G(n,p)$ . 

What is most interesting about \cref{fig:exp} is that there exists a value of $p$ between $0$ and $0.3$ for which the convergence time is maximized. When $p$ is close to $0$, the communication graph consists of many small disconnected components and convergence occurs fast. $p=1$ corresponds to the standard HK model. Somewhere in between, the time needed to converge reaches a maximum. The lesson seems to be that opinions take the longest time to converge when the probability of two agents interacting is neither too small nor too large. We also conducted the experiment when the social network was chosen from the Barab\'asi-Albert generative model for scale-free networks. The results there (not shown) are qualitatively similar.

Note that the maximizing value of $p$ seems to decrease slowly with $n$ in \cref{fig:exp}. It is not clear what the limiting behavior is as $n$ grows. Also, we currently do not have any analytical way to understand the experimental results.

\section{Non-Deterministic Hegselmann-Krause Model}
In this Section, we analyze the non-deterministic HK model. Recall that the update rule for the non-deterministic HK model is:
\begin{equation}\label{eqn:fund}
x_{t+1}(i)=x_{t}(i)+(1  + \eps_{i,t}) \frac{\sum_{j \in \cN_{t}(i)} (x_{t}(j) - x_{t}(i))}{|\cN_{t}(i)|}
\end{equation}
where $ \epsilon_{i,t} $ is an arbitrary number generated from the interval $(-\eps, \epsilon)$ at time $t$, for every $t \geq 0$. We analyze the time needed for convergence when $\eps$ is sufficiently small.

We first establish some notation. Let $\ell(t)$ be the index of the leftmost agent at time $t$. As already noted in the Introduction, $\ell(t)$ can change with $t$. Also, let $x_t(\ell)$ and $\cN_t(\ell)$  be shorthand for $x_t(\ell(t))$ and $\cN_t(\ell(t))$ respectively. Similarly, let $r(t)$ be the index of the rightmost agent at time $t$, and let $x_t(r)$ and $\cN_t(r)$ denote $x_t(r(t))$ and $\cN_t(r(t))$ respectively. 
\begin{lemma}\label{l1}
Suppose $\eps < \frac{1}{n-1}$. Then, for any agent $i \in [n]$ and for all $t \geq 0$, $x_{t+1}(i) \ge x_{t}(\ell)$. In particular, $x_{t+1}(\ell) \geq x_t(\ell)$. 
\end{lemma}
\begin{proof} 
Let $\delta =x_{t}(i)-x_{t}(\ell)$. At time $t+1$, without noise, agent $i$ can move to the left by at most $\delta (1-\frac{1}{n})$. By substituting in (\ref{eqn:fund}), we get:
$$x_{t+1}(i) \ge x_{t}(i) - (1 + \epsilon) \delta \cdot \left(1-\frac{1}{n}\right)$$
Since $\eps < \frac{1}{n-1}$, $x_{t+1}(i) \geq x_t(i) - \delta = x_t(\ell)$.
\end{proof}
\begin{remark}
By exactly the same reasoning, the position of the rightmost agent does not move to the right over time if $\eps < \frac{1}{n-1}$. 
\end{remark}

\begin{definition}For all $t \geq 0$, define the following sets:\\
\begin{align*}
L(t)&=\{i \in \cN_{t}(\ell) \mid \cN_{t}(i)=\cN_{t}(\ell) \}\\
S(t)&=\{i \in \cN_{t}(\ell) \mid \cN_{t}(i) \neq \cN_{t}(\ell) \}\\
T(t)&=\{ i \in [n] \mid  i \notin \cN_{t}(\ell) \}\\
M(t)&=S(t) \cup T(t)
\end{align*}
\end{definition}

We next show that any agent in $M(t)$ actually satisfies \cref{l1} with strict inequality. 

\begin{lemma} \label{lem:t}
For any agent $i \in T(t)$ and for all $t \geq 0$, $x_{t+1}(i) \ge x_{t}(\ell) + \frac{1}{n}-\epsilon$.
\end{lemma}
\begin{proof} Consider $i \in T(t)$ and let $k = |\cN_t(i)|$. 
\begin{align*}
 x_{t+1}(i)  & = x_{t}(i) + (1 + \epsilon_{i,t}) \sum_{j \in \cN_{t}(i)} \frac{x_{t}(j)-x_{t}(i)}{|\cN_{t}(i)|}\\
 & \ge x_{t}(i) - (1 + \epsilon)\frac{k-1}{k} \\
 & \ge x_{t}(\ell) +1-(1 + \epsilon)(1-\frac{1}{n}) \\
 & \ge x_{t}(\ell) + \frac{1}{n} - \epsilon
\end{align*}
 \end{proof}
\begin{lemma}\label{lem:s}
For any agent $i \in S(t)$ and for all $t \geq 0$, $x_{t+1}(i) \ge x_{t}(\ell) + \frac{1}{n}-\epsilon$
\end{lemma}
\begin{proof} Consider $i \in S(t)$.  Let $k= |\cN_{t}(\ell)|$ and $ \delta= x_{t}(i)-x_{t}(\ell)$.
Then agent \textit{i} moves to its left by at most:
\begin{align}
x_{t+1}(i) & \geq x_{t}(i) - ( 1 + \epsilon_{i,t} ) \frac{ \delta ( k-2 ) + 0 - (1 - \delta)}{k}  \label{eqn:chg1} \\
& \geq x_{t}(i) - ( 1 + \epsilon_{i,t} ) \frac{ \delta ( k-1 ) + 0 - (1 - \delta)}{k} \\
&= x_{t}(i) - (1 + \epsilon_{i,t}) ( \delta - \frac{1}{k} ) \nonumber \\
& \geq x_{t}(\ell) + \delta - (1 + \epsilon_{i,t}) ( \delta - \frac{1}{n} ) \nonumber \\
 & \geq x_{t}(\ell) + \frac{1}{n}- \epsilon \nonumber
 \end{align}
 \end{proof}
 Combining \cref{lem:t} and \cref{lem:s}:
 \begin{corollary}\label{c1}
For any agent $i \in M(t)$ and $t \geq 0$, $x_{t+1}(i) \ge x_{t}(\ell) + \frac{1}{n}-\epsilon$
 \end{corollary}
 \begin{lemma}\label{l5} Suppose $\eps < \frac{1}{n-1}$.  For any $t \geq 0$, if $S(t) = \emptyset$, then $\cN_{t}(\ell) = L(t) $ evolves as an independent system and for any $\rho > 0$, all the agents in $L(t)$ lie within an interval of length at most $\rho$  in time $O((\log 1/\rho)/(\log 1/\eps))$ time.
 \end{lemma}
 \begin{proof} 
 Because  $S(t) = \emptyset$, $\cN_{t}(\ell) =\cN_{t}(i) $ $ \forall \textit{i} \in \cN_{t}(\ell) $. Then, by definition,
\begin{align*}
x_{t+1}(i)
&=a_{t+1}(i) + \epsilon_{i,t} b_{t+1}(i)
\end{align*}
where $ a_{t+1}(i)=\frac{\sum_{j \in \cN_{t}(i)} x_{t}(j)}{|\cN_{t}(i)|}$ and $ b_{t+1}(i)= \frac{\sum_{j \in \cN_{t}(i)} (x_{t}(j) - x_{t}(i))}{|\cN_{t}(i)|}$
\newtheorem{obs}{Observation}
\begin{obs}
For any $i,j \in \cN_{t}(\ell)$, we have $\cN_{t}(i)=\cN_{t}(j) \implies$ $a_{t+1}(i)=a_{t+1}(j)$
\end{obs}
\begin{obs}
Since $\eps < \frac{1}{n-1}$, by \cref{l1}, for any agent $i \in \cN_{t}(\ell)$, $\cN_{t+s}(i)=\cN_{t}(i)$ for all positive integers $s$.
\end{obs}

Let $\delta_{t} := x_{t}(r)-x_{t}(\ell) \leq 1$. Then:
\begin{align*}
\delta_{t+1} 
= \epsilon_{r,t+1} b_{t+1}(r(t+1)) - \epsilon_{\ell,t+1} b_{t+1}(\ell(t+1))
\leq 2 \epsilon \cdot b_{t+1}^{\max}
\end{align*}
where
$b_{t+1}^{\max} =\max_{i}(|b_{t+1}(i)|) \leq \delta_t$.
Hence $ \delta_{t+1} \leq 2 \epsilon \delta_{t}$. Therefore, for any $\rho > 0$, all the agents will lie within an interval of length $\rho$ in time $O(\log(1/\rho)/\log(1/\eps))$. 
\end{proof}
\begin{lemma} At any time $t \geq 0$, one of the following three cases must occur:
\begin{enumerate}
\item[S1)] $S(t+1) = \emptyset$
\item[S2)] $|L(t+1)| < |L(t)| $
\item[S3)] $M(t) \cap \cN_{t+1}(\ell) \neq \emptyset$
\end{enumerate}
\end{lemma}
\begin{proof}
Assume Case S3 does not occur, meaning for all $s \in M(t)$, $s \notin \cN_{t+1}(\ell)$. Then, we show that either case S1 or S2 occurs. 

By our assumption, $\cN_{t+1}(\ell) \subseteq \cN_{t}(\ell) \setminus S(t) = L(t)$. 
Hence, $L(t+1) = \cN_{t+1}(\ell) \backslash S(t+1) \subseteq L(t) \backslash S(t+1)$.
Now, either
\begin{itemize}
\item $S(t+1) = \emptyset$, which is case S1. 
\item $S(t+1) \neq \emptyset$, and so, $L(t+1) \subseteq L(t) \setminus S(t+1) \subsetneq L(t)$, implying case S2.
\end{itemize}
\end{proof}
Note that case S1 of the above Lemma is irreversible in the sense that by \cref{c1}, each time it occurs, a subset of the agents converges independently into an interval of an arbitrarily small length $\rho$.  We next establish that case S3 can also occur only a finite number of times.

\begin{lemma}\label{lem:s3}
Suppose $\eps < \frac{1}{4n^2}$.
If $M(t) \cap \cN_{t+1}(\ell) \neq \emptyset$, then for all $i \in [n]$, $x_{t+2}(i) \geq x_{t}(\ell) + \frac{1}{4n^2}$.
\end{lemma}
\begin{proof} Suppose $s \in M(t)\cap \cN_{t+1}(\ell)$ . Then, by \cref{c1}, 
$x_{t+1}(s)-x_{t}(\ell) \geq \frac{1}{n}-\epsilon \geq \frac{1}{n}-2\epsilon$. Now, we consider two cases.  

\begin{itemize}
\item Suppose $x_{t+1}(\ell) > x_{t}(\ell) + \frac{1}{2n}-\epsilon$. 
Then, for any $i$, by \cref{l1}, $x_{t+2}(i) \geq x_{t+1}(\ell) > x_{t}(\ell) + \frac{1}{2n}-\epsilon \geq x_{t}(\ell) + \frac{1}{4n}$ ($\because \epsilon < \frac{1}{4n^2}$).

\item Otherwise, suppose $x_{t+1}(\ell) \leq x_{t}(\ell) + \frac{1}{2n} - \epsilon$. Then, $x_{t+1}(s) - x_{t+1}(\ell) \geq \frac{1}{2n} - \epsilon$.
Now for any agent $i \in L(t+1)$, we have:
\begin{align}
x_{t+2}(i)  
&= x_{t+1}(i) + (1 + \epsilon_{i,t+1}) \sum_{j \in \cN_{t+1}(i)} \frac{x_{t+1}(j)-x_{t+1}(i)}{|\cN_{t+1}(i)|} \label{eqn:chg2}\\
&=\sum_{j \in \cN_{t+1}(\ell)} \frac{x_{t+1}(j)}{|\cN_{t+1}(l)|} + \epsilon_{i,t+1} \sum_{j \in \cN_{t+1}(i)} \frac{(x_{t+1}(j) - x_{t+1}(i))}{|\cN_{t+1}(i)|} && (\because i \in L(t+1)) \nonumber \\
&\ge x_{t+1}(\ell) + \frac{x_{t+1}(s)-x_{t+1}(\ell)}{n} - \epsilon (1-\frac{1}{n}) \nonumber\\
&\ge x_{t+1}(\ell) + \frac{(\frac{1}{2n}-\epsilon)}{n}-\epsilon (1-\frac{1}{n}) \nonumber\\
&\ge x_{t+1}(\ell) + \frac{1}{2n^2} - \epsilon \nonumber
\end{align}
If $\epsilon < \frac{1}{4n^2}$, then $x_{t+2}(i) > x_{t+1}(\ell) + \frac{1}{4n^2} \geq x_{t}(\ell) + \frac{1}{4n^2}$, by \cref{l1}. If $i \in M(t+1)$, then by \cref{c1}, we have $ x_{t+2}(i) \geq x_{t+1}(\ell) + \frac{1}{2n} \geq x_{t}(\ell) + \frac{1}{2n}$, the second inequality again due to \cref{l1}. So, the claim is proved for all $i$.
\end{itemize}
\end{proof}
We are now ready to prove our main theorem.\\

\noindent 
\textbf{\cref{thm:ndmain} (recalled)}
{\em Suppose $\eps < \frac{1}{4n^2}$.  Starting with an arbitrary initial arangement of $n$ agents evolving according to  the non-deterministic HK model, the number of steps before which all the agents are confined within intervals of length $\rho$ thereafter is $O(n^4 + \log(1/\rho)/\log n)$.}
\begin{proof} Since for any $t$, $|L(t)| \leq n$, case S2 can occur consecutively at most $n$ times.  So, within every $n$ time steps, case {S1} or  case {S3} must occur at least once. Case S1 can clearly occur at most $n$ times, whereas by \cref{lem:s3}, case S3 can occur $O(n^3)$ times. Hence,  after time $O(n^4)$ time steps, all agents lie in independent subsystems, each of diameter at most $1$. Each of these subsystems, by \cref{l5}, cluster into intervals of length at most $\rho$ in $O(\log(1/\rho)/\log(1/\eps)) = O(\log(1/\rho)/\log n)$ time steps.
\end{proof}

\begin{remark}\label{rem:ndgen}
Suppose we change the definition of non-deterministic HK models so that each agent is influenced non-uniformly by its neighbors. Specifically, let the update rule be:
\begin{equation}
x_{t+1}(i) = x_t(i) + \frac{\sum_{j \in \cN_t(i)} (1+\eps_{i,j,t}) (x_t(j) - x_t(i))}{|\cN_t(i)|} \label{eqn:nd2}
\end{equation}
where each $\eps_{i,j,t}$ is an arbitrary\footnote{perhaps generated endogenously} number generated from the interval $(-\eps, \eps)$. Most of the above proof needs no modification. (\ref{eqn:chg1}) changes to:
\begin{align*}
x_{t+1}(i) & \geq x_{t}(i) - \frac{ ( 1 + \eps)\delta ( k-2 ) + 0 - (1- \eps) (1 - \delta)}{k}  \\
& \geq x_{t}(i) - \frac{ ( 1 + \eps)\delta ( k-1 ) + 0 - (1- \eps) (1 - \delta)}{k} \\
& = x_{t}(\ell) +\frac{1}{k} -\eps \left(\delta -2\delta/k  + 1/k\right)\\
&\geq x_t(\ell) + \frac{1}{n} - 	2\eps \geq x_t(\ell) + \frac{1}{2n}
\end{align*}
if $\eps < \frac1{4n^2}$. 
Everywhere else, the claims hold straighforwardly using the upperbound of $\eps$ on each $\eps_{i,j,t}$. Hence, for this system also, \cref{thm:ndmain} holds.

An interesting aspect about (\ref{eqn:nd2}) is that an agent might move in the opposite direction than it would move in the classical HK model, whereas in (\ref{eqn:fund}), the direction of each agent's movement is the same as in classical HK.
\end{remark}

\section{Future Directions}
There are a number of open directions suggested by the problems studied in this work.
\begin{itemize}
\item
In our formulation of the social HK model, we required the underlying social network to be undirected. This leads to bidirectional dynamical systems. What happens if the social network is directed? 
  Proving convergence for the HK  model with a directed social network seems quite challenging because it includes, as a special case, the ill-understood HK model with stubborn agents (i.e., there are some agents with confidence bound 0 while all others have confidence bound 1). To see this, let every non-stubborn agent have edges to all agents and every stubborn agent have no outgoing edges. 
  
 \item
 We introduced the notion of friendly social HK systems in \cref{sec:friendly} and showed that these allow only a polynomial number of non-trivial steps. We conjecture that friendliness is necessary for a polynomial bound. Can one demonstrate a non-friendly HK system for which there are an exponential number of non-trivial steps? Is friendliness indeed a tight condition for a polynomial bound?

\item
Is there a rigorous justification for the empirical results reported in \cref{sec:exp}? In general, it would be interesting to understand the effect of the social network structure on the dynamics of the HK model. 

\item
For the non-deterministic HK model, it is important to increase the range of $\eps$ for which \cref{thm:ndmain} is valid. Note that if $\eps$ is allowed to be in $[-1,0]$, then we could prove convergence for the HK system with stubborn agents (by simply setting $\eps_{i,t} = -1$ for all $t$ if agent $i$ is stubborn and $\eps_{i,t}=0$ otherwise). Moreover, in the generalized non-deterministic HK systems described in \cref{rem:ndgen}, if $\eps$ is allowed to be in $[-1,n-1]$, then we can simulate arbitrary heterogeneous HK systems (by setting $\eps_{i,j,t} = -1$ if $j \notin \cN_t(i)$ and $\eps_{i,j,t} = \frac{|\{k : |x_t(k) - x_t(i)| \leq 1\}|}{|\cN_t(i)|}-1$ otherwise). 

\item
Can we prove a polynomial bound for the convergence of the non-deterministic HK model in multiple dimensions? Our current proof does not extend while the approach used in \cref{sec:soc} seems sensitive to the presence of non-determinism.
\end{itemize}

\section*{Acknowledgments}
We thank Ashish Goel for very helpful discussions and Vinay Vasista for assisting with the experiments reported in \cref{sec:exp}. 

\bibliographystyle{alpha}
\bibliography{papers}

\end{document}